\documentclass[aip,amsmath,amssymb,reprint,nofootinbib]{revtex4-1}

\usepackage[accsupp]{axessibility}
\usepackage[utf8]{inputenc}
\usepackage{graphicx}
\usepackage{bbm}
\usepackage{subcaption}
\usepackage{amsthm}
\usepackage{hyperref}

\newtheorem{lemma}{Lemma}
\newtheorem{remark}{Remark}
\newtheorem{theorem}{Theorem}

\begin{document}

\preprint{AIP/123-QED}

\title{Filtering Dynamical Systems Using Observations of Statistics}

\date{\today}

\author{Eviatar Bach}
\email{eviatarbach@protonmail.com}
\affiliation{Department of Environmental Science and Engineering, California Institute of Technology, Pasadena, California 91125, USA}
\affiliation{Department of Computing and Mathematical Sciences, California Institute of Technology, Pasadena, California 91125, USA}
\author{Tim Colonius}
\affiliation{Department of Mechanical and Civil Engineering, California Institute of Technology, Pasadena, California 91125, USA}
\author{Isabel Scherl}
\affiliation{Department of Mechanical and Civil Engineering, California Institute of Technology, Pasadena, California 91125, USA}
\author{Andrew Stuart}
\affiliation{Department of Computing and Mathematical Sciences, California Institute of Technology, Pasadena, California 91125, USA}

\begin{abstract}
    We consider the problem of filtering dynamical systems, possibly stochastic, using observations of statistics. Thus, the computational task is to estimate a time-evolving density $\rho(v, t)$ given noisy observations of the true density $\rho^\dagger$; this contrasts with the standard filtering problem based on observations of the state $v$. The task is naturally formulated as an infinite-dimensional filtering problem in the space of densities $\rho$. However, for the purposes of tractability, we seek algorithms in state space; specifically, we introduce a mean-field state-space model, and using interacting particle system approximations to this model, we propose an ensemble method. We refer to the resulting methodology as the ensemble Fokker--Planck filter (EnFPF).
    
    Under certain restrictive assumptions, we show that the EnFPF approximates the Kalman--Bucy filter for the Fokker--Planck equation, which is the exact solution to the infinite-dimensional filtering problem. Furthermore, our numerical experiments show that the methodology is useful beyond this restrictive setting. Specifically, the experiments show that the EnFPF is able to correct ensemble statistics, to accelerate convergence to the invariant density for autonomous systems, and to accelerate convergence to time-dependent invariant densities for non-autonomous systems.  We discuss possible applications of the EnFPF to climate ensembles and to turbulence modeling.
\end{abstract}

\maketitle

\begin{quotation}
Data assimilation (DA) is the process of estimating the state of a dynamical system using observations. Here, we modify the standard DA setting to allow for observations of \emph{statistics} of a system with respect to its time-evolving probability density. We propose a mathematical framework, a resulting ensemble method, and present numerical experiments demonstrating accelerated convergence of a system to its attractor. We propose further applications to problems in climate and turbulence modeling.
\end{quotation}

\section{Introduction}

The goal of this paper is to introduce a filtering methodology that incorporates statistical information into a (possibly stochastic) dynamical system. In Secs. \ref{ssec:1.1}--\ref{ssec:1.3}, we present, respectively, a high-level overview of the problem, discuss the motivation and previous literature, and outline the paper structure and our contributions.

\subsection{Assimilating Statistical Observations}
\label{ssec:1.1}

We start by presenting a high-level overview of the problem of incorporating statistical information into a dynamical system; a detailed problem statement follows in Sec. \ref{ssec:problem}.

{Data assimilation (DA)} is overviewed in a number of books, including Refs.\cite{jazwinski_stochastic_1970,kalnay_atmospheric_2002,law_data_2015,reich_probabilistic_2015}. The problem is to estimate the state of a dynamical system by combining noisy, partial observations with a model for the system. In the continuous-time DA problem, we have a stochastic differential equation (SDE)
\begin{align}
    dv^\dagger &= f(v^\dagger, t)\,dt + \sqrt{\Sigma(t)}\,dW,\label{eq:sde}\\
    {v^\dagger(0)} &{= v^\dagger_0,}
\end{align}
with solution $v^\dagger\in\mathbb{R}^d$, and observations given by
\begin{equation}
    dz^\dagger = h(v^\dagger(t), t)\,dt + \sqrt{\Gamma(t)}\,dB,\label{eq:obs_standard}
\end{equation}
with $z^\dagger\in\mathbb{R}^p$. The equations for $v^\dagger$ and $z^\dagger$ are driven by independent standard Wiener processes $W$ and $B$. These SDEs, as with all the SDEs in the paper, are to be interpreted in the It\^o sense. Filtering is then the problem of obtaining the best possible estimate of the posterior density on $v^\dagger(t)$ given the past observations $\{z^\dagger(s)\}_{s\in[0,t]}$. Throughout the paper, we use the $\dagger$ superscript to indicate the true quantities, and omit it for filtered quantities.

Instead of observing a specific trajectory of a dynamical system, as $\{z^\dagger(t)\}$ given by Eq.~\eqref{eq:obs_standard} does, one can also consider \emph{observations of the system's statistical behavior}, that is, observations of functionals of the probability density $\rho^\dagger(v,t)$ over trajectories. This density reflects the randomness from the initial conditions for $v$ and/or from the Brownian forcing. For a deterministic dynamical system ($\Sigma \equiv 0$), if the initial conditions are random, then $\rho^\dagger(v,t)$ will reflect the changing density over time under the action of the system's dynamics, governed by the Liouville equation.\footnote{Here, we use the term Liouville equation for the equation governing evolution of the density of any ordinary differential equation, not just in the Hamiltonian setting.} If noise is present, the changing density is also affected by the Brownian noise $W$ and is governed by the Fokker--Planck equation, a diffusively regularized Liouville equation. In this paper we focus on observations of $\rho^\dagger(v, t)$ defined by replacing Eq.~\eqref{eq:obs_standard} with
\begin{equation}
    dz^\dagger = \left(\int \mathfrak{h}(v, t) \rho^\dagger(v, t) \,dv\right) dt + \sqrt{\Gamma(t)}\,dB.\label{eq:obs_stat}
\end{equation}
Here, $\mathfrak{h}(v, t)$ defines the observed statistics  of $v$, $B$ is a Wiener process, and $z^\dagger\in\mathbb{R}^p$. The filtering problem is to estimate {a density} $\rho(v, t)$ given all the past observations $\{z^\dagger(s)\}_{s\in[0,t]}$. As in the observation equation \eqref{eq:obs_standard}, the observations are finite-dimensional, noisy, and partial. However, since the observations are now of $\rho^\dagger(v, t)$ instead of $v^\dagger(t)$, we must specify the dynamics of $\rho^\dagger(v, t)$. This is given by the Fokker--Planck (FP) or Kolmogorov forward equation,
\begin{subequations}
\begin{align}
    \frac{\partial \rho^\dagger}{\partial t} &= \mathcal{L}^*(t)\rho^\dagger,\\
    \mathcal{L}^*(t)\psi &= -\nabla\cdot(\psi f) + \frac{1}{2} \nabla\cdot\bigl(\nabla\cdot(\psi \Sigma)\bigr),
    \end{align}\label{eq:fp}\end{subequations}
where $\mathcal{L}^*$ is the adjoint of the generator of Eq.~\eqref{eq:sde}.\footnote{We define the divergence of a matrix as is standard in continuum mechanics; see Gurtin (1981)\cite{gurtin_introduction_1981} and Gonzalez and Stuart (2008)\cite{gonzalez_first_2008}. The divergence of a matrix $S$ is defined by the identity $(\nabla\cdot S)\cdot a = \nabla\cdot(S^T a)$ holding for any vector $a$.}. For a deterministic system, with $\Sigma \equiv 0$, the Fokker--Planck equation reduces to the Liouville equation.

An important question is how one would obtain observations of a system's statistics for problems of practical relevance. We discuss this in detail in Sec. \ref{sssec:obtaining_obs}. For now, we proceed on the assumption that $z^\dagger$ solving Eq.~\eqref{eq:obs_stat} is given.

Now, Eqs.~\eqref{eq:fp} and \eqref{eq:obs_stat} define a filtering problem for $\rho(v, t)$. This is an infinite-dimensional filtering problem, in contrast to the finite-dimensional filtering problem for $v(t)$ defined by Eqs.~\eqref{eq:sde} and \eqref{eq:obs_standard}. We refer to the filtering problem defined by Eqs.~\eqref{eq:fp} and \eqref{eq:obs_stat} as the \emph{Fokker--Planck filtering problem}. Note that both Eqs.~\eqref{eq:fp} and \eqref{eq:obs_stat} are \emph{linear} in $\rho^\dagger$, meaning that the solution to the problem can be written using the infinite-dimensional Kalman--Bucy (KB) filter; see Sec. \ref{ssec:kb} for more details. 

{Despite the existence of an exact solution to the filtering problem,
through the infinite-dimensional Kalman--Bucy (KB) filter, approximating the Gaussian conditional density $\rho$ is in most setting computationally intractable since the mean is a probability density function and the covariance is an operator. Thus we} seek inspiration from the success of ensemble Kalman filtering \cite{evensen_sequential_1994}: we work in state space and seek an ensemble that evolves in time a number of states whose empirical density approximates the filtered $\rho$. We note that the particle filter similarly substitutes the problem of evolving a probability density with that of evolving a number of particles and weights \citep{crisan_interacting_1999}. Furthermore, derivation of ensemble Kalman methods via a mean-field limit provides a systematic methodology for the derivation of equal-weight approximate filters \citep{calvello_ensemble_2022}. We call the resulting method the \emph{ensemble Fokker--Planck filter} (EnFPF).

\begin{figure}
    \centering
    \includegraphics[scale=0.45]{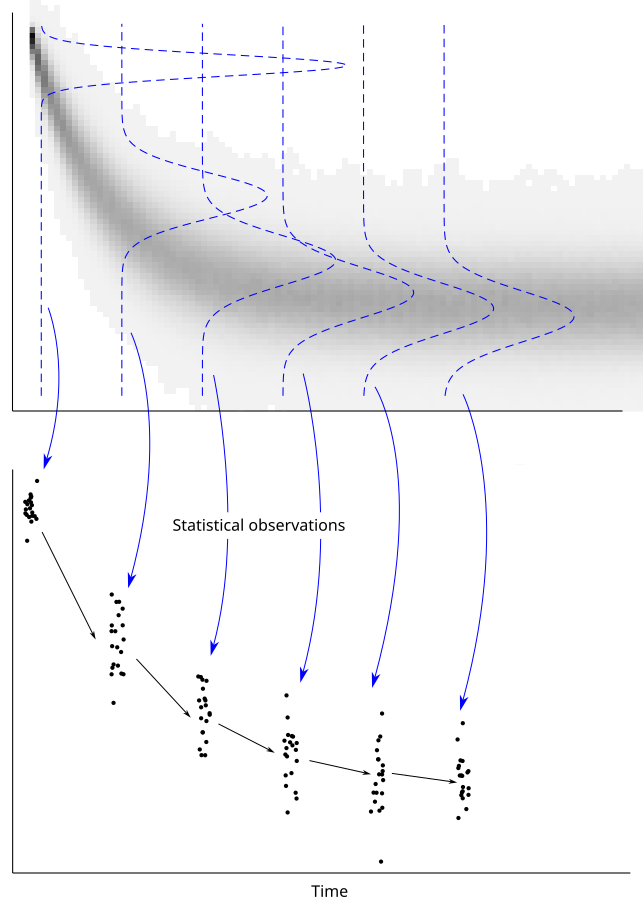}
    \caption{The density of an Ornstein--Uhlenbeck process evolving in time (top panel). At regular intervals, we make observations of this density and use them to inform the evolution of an ensemble (bottom panel).}
    \label{fig:schematic}
\end{figure}

Figure~\ref{fig:schematic} shows a schematic of such an ensemble method. In the top panel is the true time-varying probability density, in this case of an Ornstein--Uhlenbeck process. In the bottom panel is an ensemble of states. At regular intervals, we observe expectations over the density in the top panel. {Using these observations and our model of the system, we evolve the ensemble over the time interval between the current and next observations.}

\subsection{Motivation and Literature Review}\label{ssec:motivation}

{The subject of Kalman filtering and Kalman--Bucy (KB) filtering in infinite-dimensional spaces is studied in the control theory literature \citep{curtain_infinite_1978}. 
We emphasize that although we sketch out the basic mathematical foundations of the Fokker--Planck filtering problem in Sec. \ref{sec:justify}, many interesting mathematical problems in analysis and probability remain open in this area.
To the best of our knowledge, the methodology proposed here is the first general method for assimilating observations of statistics directly into a state-space formulation of dynamical systems. Our methodology is built on the conceptual approach introduced in the feedback particle filter \citep{yang_feedback_2013,reich_data_2019}, and earlier related work \citep{crisan_approximate_2010}, seeking a mean-field model that achieves the goal of filtering and can be approximated by particle methods \citep{pathiraja_mckean--vlasov_2021}; in particular, we seek particle approximations of the mean-field model inspired by ensemble Kalman methods \citep{calvello_ensemble_2022}.

The problem of recovering a probability density from a finite number of known moments is called a moment problem. When $\mathfrak{h}$ in Eq.~\eqref{eq:obs_stat} consists of monomials in $v$, the problem of reconstructing $\rho$ is similar to a moment problem, with the major difference that $\rho$ evolves in time according to a dynamical system. Moment problems are typically regularized by a maximum entropy approach \citep{abramov_multidimensional_2010}; in the Fokker--Planck filtering problem, regularization is provided by the system's dynamics.

Our motivation comes from a number of applications around which we organize the remainder of our literature review, after first discussing the general question of how to obtain observations of statistics.}

{
\subsubsection{Obtaining observations of statistics}\label{sssec:obtaining_obs}

In typical applications, one can only observe a single trajectory of a dynamical system, and thus the statistics of the density will not be directly available. If we are interested in the statistics of the invariant measure, as we are for several of the applications discussed below, then for ergodic systems we have that
\begin{equation}
    \lim_{T\to\infty}\frac{1}{T}\int_0^T \mathfrak{h}(v^\dagger(t))\,dt = \int \mathfrak{h}(v)\rho^\dagger(v)\,dv,
\end{equation}
where $\rho^\dagger$ is the invariant density, and thus an approximation of the statistics of the invariant measure can be obtained from a long observed or simulated trajectory of the dynamical system.

For nonautonomous systems, due to lack of ergodicity, observations of the statistics cannot be made using long time averages. If the nonstationary forcing is slow enough, however, an adiabatic approximation, in which the fast scales are considered to be ergodic with an invariant measure parameterized by the value of the slow forcing, may be justified \citep{pavliotis_multiscale_2008,drotos_quantifying_2016}. If the forcing is periodic, then observations of the phase-dependent statistics could be obtained by averaging the observables at a given phase over multiple periods.

For certain systems, invariant statistics may be acquired analytically or by numerically solving a different set of equations. For example, for the Navier--Stokes equations, the Reynolds-averaged Navier--Stokes (RANS) equations can be used to approximate the stationary statistics.

It may be possible to instead formulate a filtering problem using an observation operator that involves averaging over a finite time window; we leave this for future work. This problem was considered in Ref. \cite{dirren_toward_2005}, but only a heuristic solution was proposed. We note that other works have made use of observation operators with time-delayed observations \cite{rey_using_2014,bach_ensemble_2021}, albeit for different purposes.}

In Secs. \ref{ssec:IB2}--\ref{ssec:IB5}, we review the possible applications of the ensemble Fokker--Planck filter.

\subsubsection{Acceleration of convergence to a (possibly time-dependent) invariant measure}\label{ssec:IB2}
Acceleration of the time to convergence of dynamical models to their invariant measure (often referred to as the ``spin-up'' period or the transient) is of importance in many fields, including climate \citep{bryan_accelerating_1984,daron_quantifying_2015,drotos_importance_2017,deser_insights_2020} and other fluid problems \citep{wang_deepparticle_2022}, Langevin sampling \citep{hwang_accelerating_2005,abdulle_accelerated_2019}, and turbulence simulation \citep{nelson_reducing_2017}.

For a stochastic differential equation with an invariant measure, under conditions described in Goldys and Maslowski (2005)\cite{goldys_exponential_2005}, the convergence to this invariant measure is exponential with an exponent related to the spectral gap of the corresponding generator.

In this paper, we show that this convergence can be accelerated using the ensemble Fokker--Planck filter, and this is the primary application we test in the numerical experiments. In particular, if some statistics of the invariant measure are known, these statistics can be assimilated into the ensemble, obtaining an ensemble whose empirical density is closer to the invariant measure.

To our knowledge, existing methods of accelerating convergence {of model trajectories} to the invariant measure have been problem-dependent, as in Bryan (1984)\cite{bryan_accelerating_1984}. Isik (2013)\cite{isik_spin_2013} and Isik, Takhirov, and Zheng (2017)\cite{isik_second_2017} studied a relaxation-based method of accelerating the convergence to equilibrium of the Navier--Stokes equations, which bears some resemblance to our approach.

Non-autonomous (also referred to as non-stationary) and random dynamical systems can have time-dependent attractors, known as pullback attractors, to which the evolution converges \cite{arnold_random_1998}. A pullback attractor is the set that the dynamical system approaches when evolved in time from the infinite past to a fixed time (say time $0$ without loss of generality). We refer to the probability measure associated with these attractors as time-dependent invariant measures, following Chekroun, Simonnet, and Ghil (2011)\cite{chekroun_stochastic_2011}. These objects are of considerable interest for climate \citep{ghil_climate_2008,chekroun_stochastic_2011,daron_quantifying_2015}. The EnFPF can also accelerate convergence to these invariant measures.

The problem of accelerating convergence to the invariant measure is related to the problem of {controlling the Fokker--Planck equation, where a density is controlled in order to reach to a specified target distribution\cite{annunziato_fokkerplanck_2013}, and to} statistical control, wherein one aims to return a perturbed system to its equilibrium statistics\cite{covington_effective_2023}.

Furthermore, the EnFPF could be tested for accelerating the convergence of sampling algorithms such as Langevin sampling and Markov chain Monte Carlo, when some statistics of the target density are known \emph{a priori}.

Finally we note that when estimating Koopman or Perron--Frobenius operators, it is often necessary to have a large number of trajectories from initial conditions sampled from the invariant measure.

\subsubsection{Parameter estimation}
The EnFPF could be used for jointly updating states and parameters using statistical observations by adopting a state augmentation approach.
Other work has adapted methods from data assimilation for parameter estimation using time-averaged statistics, assumed to be close to the statistics on the invariant measure by ergodicity \citep{kohl_adjoint_2002,schneider_learning_2021,mons_ensemble-variational_2021,schneider_ensemble_2022}.

\subsubsection{Correcting for model error} Generally, methods that correct for model error are formulated in terms of forecast performance at some lead time, e.g., \cite{levine_framework_2022-1} and \cite{bach_multi-model_2023}. If one is instead interested in correcting statistical properties, one can postulate a parametric form for the model error and use time-averaged observations to estimate the parameters, as discussed in the preceding paragraph. Alternatively, the EnFPF could be tested for directly correcting model error using statistical observations in a similar manner to the use of classical DA in reducing the impact of model error for forecast applications \cite{danforth_using_2008,chen_correcting_2022}. The analysis increments could then be taken to approximate model error corrections, and training a machine learning model to predict these corrections could be tested, as has been done for classical DA \citep{danforth_using_2008,farchi_using_2021,chen_correcting_2022}.

Statistical properties have previously been used to learn closure models for the Navier--Stokes equation using a 3DVar-like scheme \cite{ephrati_data-assimilation_2023}.

\subsubsection{Assimilation of time-averaged observations}\label{ssec:IB5}
In paleoclimate, proxy records often represent time averages instead of instantaneous measurements. Methods have been developed for making use of time-averaged observations for state estimation in the paleoclimate data assimilation literature \citep{dirren_toward_2005}. As discussed above, in the case of slow forcing, time-averaged observations can be used to approximately track the system's time-varying statistics, enabling their use in the EnFPF.

\subsection{Contributions and Paper Outline}
\label{ssec:1.3}

The primary contributions of this work are: (i) to establish a framework for the filtering of stochastic dynamical systems, or dynamical systems with random initial data, given only observations of statistics; (ii) to introduce ensemble-based state-space methods for this filtering problem via a mean field perspective; and (iii) to demonstrate numerically that the proposed methods are effective at guiding dynamical systems toward observed statistics. (i) is covered in Sec. \ref{ssec:problem} and Sec. \ref{sec:justify}; (ii) is covered in sections \ref{ssec:mean-field_intro}--\ref{ssec:implementation}; and (iii) is covered in Sec. \ref{sec:experiments}.

In Sec. \ref{ssec:problem}, we outline the Fokker--Planck filtering problem and distinguish it from the standard filtering problem. In Secs. \ref{ssec:mean-field_intro}--\ref{ssec:discrete-time}, we introduce a mean-field algorithm and its particle and discrete-time approximations, culminating in the ensemble Fokker--Planck filter (EnFPF). In Sec. \ref{ssec:implementation}, we discuss implementation details, including the approximation of the score function and a square-root ensemble formulation with reduced computational effort. 

In Sec. \ref{sec:experiments}, we carry out numerical experiments with several chaotic dynamical systems, both autonomous and non-autonomous, and based on the Lorenz63, Lorenz96, and Kuramoto--Sivashinsky models. In particular, we demonstrate that the EnFPF can accelerate the convergence of these systems to their invariant densities, using information about the moments of these densities.

In Sec. \ref{sec:justify}, we provide a justification of our algorithm. We first formulate the KB filter for densities (Sec. \ref{ssec:kb}), which provides a solution to the Fokker--Planck filtering problem in function space, and analyze some of its properties in Appendix~\ref{sec:kb_properties}. We then propose an ansatz amenable to a mean-field model (Sec. \ref{ssec:ansatz}) and show its equivalence to the KB filter for densities under some assumptions (Theorem \ref{theorem:equiv} in Appendix \ref{appendix:proof}). We then show how this ansatz can be approximated by a mean-field model (Sec. \ref{sec:mean_field_intro}, providing further details in Appendix~\ref{ssec:mean-field}).

Finally, in Sec. \ref{ssec:conclusions}, we give conclusions and outlook for future work.

\section{Problem and Algorithm}\label{sec:algorithm}

In Sec. \ref{ssec:problem}, we introduce the probabilistic formulation of the standard filtering problem, and then contrast it with the Fokker--Planck filtering problem,  
where data is in the form of statistics. Sec. \ref{ssec:mean-field_intro} demonstrates an approach to this problem using a mean-field model. In Sec. \ref{ssec:particle_approx}, we introduce a particle approximation of the mean-field algorithm, which forms the basis of the proposed EnFPF.

\subsection{Problem Statement}\label{ssec:problem}

\subsubsection{The Standard Filtering Problem}

In the standard filtering problem, we are given state observations $z^\dagger(t)$ of $v^\dagger(t)$, defined by Eq.~\eqref{eq:obs_standard}, and the dynamics of $v^\dagger(t)$ are given by Eq.~\eqref{eq:sde}. The problem is then to find an equation for the conditional distribution of $v | Z^\dagger(t)$, where $Z^\dagger(t) = \{z^\dagger(s)\}_{s\in[0,t]}$ are the observations accumulated up to time $t$ under a fixed realization of $B$. The solution to the filtering problem is given by the Kushner--Stratonovich equation,
\begin{equation}
    \frac{\partial \rho}{\partial t} = \mathcal{L}^*{(t)}\rho + \Bigl\langle h(v,{t}) - \mathbb{E}h, \frac{dz^\dagger}{dt} - \mathbb{E}h \Bigr\rangle_{\Gamma{(t)}}\rho,\label{eq:ks}
\end{equation}
where $\langle\cdot,\cdot\rangle_A\equiv \langle A^{-1/2}\cdot,A^{-1/2}\cdot\rangle$ is the weighted Euclidean inner product. Treatments of the standard filtering problem can be found in Jazwinski (1970)\cite{jazwinski_stochastic_1970} and Bain and Crisan (2009)\cite{bain_fundamentals_2009}.

\subsubsection{The Fokker--Planck Filtering Problem}

In this paper, we consider instead noisy observations of $\rho^\dagger(v, t)$: the observation process $z^\dagger(\cdot)$ is given by
\begin{equation}
    dz^\dagger = H(t)\rho^\dagger(\cdot,t)\,dt + \sqrt{\Gamma(t)}\; dB.\label{eq:fp_obs}
\end{equation}
Here, $H(t)$ is a linear operator mapping the space of probability densities into a finite-dimensional Euclidean space, and the dynamics of $\rho^\dagger$ are given by the Fokker--Planck equation \eqref{eq:fp}. That is, we make observations of \emph{statistics} of the dynamical system. We refer to the problem of finding the conditional density of $v | Z^\dagger(t)$, where $Z^\dagger(t) = \{z^\dagger(s)\}_{  s\in[0,t]}$ is given by Eq.~\eqref{eq:fp_obs}, as the \emph{Fokker--Planck filtering problem}. In the following subsection, we propose an approximation to the solution to this problem in state space.

\subsection{Mean-Field Equation}\label{ssec:mean-field_intro}

Although in Sec. \ref{ssec:kb} we treat the Fokker--Planck filtering problem for more general $H$, in the rest of what follows {we focus on the setting where}
\begin{equation}
    H(t)\rho = \mathbb{E}[\mathfrak{h}(v, t)] = \int \mathfrak{h}(v, t) \rho(v, t) \,dv,\label{eq:H_def}
\end{equation}
for some $\mathfrak{h}$. With this assumption on $H$, Eq.~\eqref{eq:fp_obs} reduces to Eq.~\eqref{eq:obs_stat}. In particular, if $\mathfrak{h}$ is a monomial in $v$, e.g., $\mathfrak{h}(v) = v$ or $\mathfrak{h}(v) = \operatorname{vec}(v\otimes v)$, then $H\rho$ will correspond to moments of $\rho$. We will henceforth use $\mathbb{E}$ to denote expectation under $\rho$, unless otherwise indicated.

\begin{remark}Note that if $\rho^\dagger(v, 0) = \delta(v - v^\dagger_0)$ for some $v^\dagger_0$ and $\Sigma = 0$, then the Fokker--Planck filtering problem is equivalent to the standard filtering problem with $v^\dagger(0) = v^\dagger_0$, observation operator $\mathfrak{h}$, and $\Sigma = 0$.\end{remark}

Our proposed methodology is to introduce a mean-field model for variable $v$, depending on its own probability density function $\rho(v,t)$. The mean-field model is chosen to drive the system toward the observed statistical information. Algorithms are then based on particle approximation of this model, leading to ensemble Kalman--type methods. The mean-field model considered is
\begin{subequations}
\label{eq:mean_field}
\begin{align}
    dv &= f(v, t) \, dt + \sqrt{\Sigma(t)}\, dW + K(t)\bigl(dz^\dagger - d\hat{z}\bigr),\\
    d\hat{z} &= (\mathbb{E}\mathfrak{h})(t)\, dt + \sqrt{\Gamma(t)} dB,\\
    K(t) &= C^{v\mathfrak{h}}(t)\Gamma(t)^{-1},\\
    C^{v\mathfrak{h}}(t) &= \mathbb{E}\bigl[\bigl(v(t) - \mathbb{E}v(t)\bigr)\bigl(\mathfrak{h}(v,t) - (\mathbb{E}\mathfrak{h})(t)\bigr)^T\bigr].
\end{align}
\end{subequations}
The terms in the mean-field model can be understood intuitively as follows. The first two terms on the right-hand side of Eq.~(\ref{eq:mean_field}a) are simply the dynamics of the system \eqref{eq:sde}. The third term resembles the standard nudging observer term from control theory, with an ensemble Kalman--inspired gain, and the use of noisy simulated data, as in the stochastic ensemble Kalman filter.

In some problems, we find that it is beneficial to include an additional score-based term in the model, replacing Eq.~(\ref{eq:mean_field}a) by 
\begin{align}
\label{eq:mean_field_K2}
    dv &= f(v, {t}) \, dt + \sqrt{\Sigma{(t)}}\, dW + K{(t)}\Bigl(dz^\dagger - d\hat{z}\Bigr)\nonumber\\
    &\quad\quad\quad\quad\quad\quad\quad\quad+ K{(t)}\Gamma{(t)} K{(t)}^T \nabla \log \rho{(v, t)}\,dt.
\end{align}
The additional term induces negative diffusion in the equation for the density of $v$, exactly balancing the diffusion introduced through $z^\dagger$ and $\hat{z}$. We justify equations \eqref{eq:mean_field} and \eqref{eq:mean_field_K2} in detail in Sec. \ref{sec:justify} by building on the Fokker--Planck picture in density space.

\subsection{Particle Approximation of Mean-Field Equation}\label{ssec:particle_approx}

In order to tractably implement the mean-field equations \eqref{eq:mean_field}, we use a particle (or ensemble) approximation. That is, given $J$ particles, we consider the following interacting particle system for 
$\{v^{(j)}\}_{j=1}^{\mathsf{J}}$:
\begin{subequations}
\label{eq:firstsub}
\begin{align}
    dv^{(j)} &= f(v^{(j)}, t)\,dt + \sqrt{\Sigma(t)}\, dW^{(j)} + K(t)\bigl(dz^\dagger - d\hat{z}^{(j)}\bigr),\\
    d\hat{z}^{(j)} &= (\mathbb{E}^\mathsf{J}\mathfrak{h})(t)\,dt + \sqrt{\Gamma(t)}\,dB^{(j)},\label{eq:pred_obs}\\
    K(t) &= (C^{v\mathfrak{h}}(t))^{\mathsf{J}}\,\Gamma(t)^{-1}.
\end{align}
\end{subequations}
Here, $\mathbb{E}^\mathsf{J}$ denotes expectation with respect to the empirical measure formed by equally weighting Dirac measures at the particles $\{v^{(j)}\}_{j=1}^{\mathsf{J}}$; $(C^{v\mathfrak{h}})^{\mathsf{J}}$ denotes the sample cross-covariance computed using this empirical measure,
\begin{equation*}
C^{v\mathfrak{h}}(t) = \mathbb{E}^\mathsf{J}\bigl[\bigl(v(t) - \mathbb{E}v(t)\bigr)\bigl(\mathfrak{h}(v,t) - (\mathbb{E}\mathfrak{h})(t)\bigr)^T\bigr].
\end{equation*}

Note that unlike the ensemble Kalman filter, the predicted observation for each ensemble member, Eq.~\eqref{eq:pred_obs}, involves the expectation of $\mathfrak{h}$ over the ensemble, instead of the observation operator applied to that ensemble member.

\subsection{Discrete-Time Approximation of Mean-Field Equation}\label{ssec:discrete-time}

A discrete-time analog of Eqs.~\eqref{eq:firstsub} is given by
\begin{subequations}
\label{eq:discrete}
\begin{align}
    \hat{v}^{(j)}_{i+1} &= \Psi_i(v_i^{(j)}) + \xi^{(j)}_i,\\
    v^{(j)}_{i+1} &= \hat{v}^{(j)}_{i+1} + K_{i+1}(y_{i+1}^\dagger - \hat{y}^{(j)}_{i+1}),\\
    \hat{y}^{(j)}_{i+1} &= \mathbb{E}^{\mathsf{J}}[\mathfrak{h}_{i+1}(\hat{v}_{i+1})] + \eta^{(j)}_{{i+1}},\\
    K_{i+1} &= (\hat{C}_{i+1}^{v\mathfrak{h}})^\mathsf{J}((\hat{C}_{i+1}^{\mathfrak{h}\mathfrak{h}})^\mathsf{J} + (\Gamma_d)_{i+1})^{-1},
\end{align}
\end{subequations}
where $\xi^{(j)}_i \sim \mathcal{N}(0, (\Sigma_d)_i)$, $\eta^{(j)}_i \sim \mathcal{N}(0, (\Gamma_d)_i)$, $\mathfrak{h}_i(v) = \mathfrak{h}(v, t)$, and
\begin{align}
    (\hat{C}_{i+1}^{v\mathfrak{h}})^\mathsf{J} &= \mathbb{E}^{\mathsf{J}} [(\hat{v}_{i+1} - \mathbb{E}^{\mathsf{J}}\hat{v}_{i+1})\\
    &\quad\quad\quad\otimes (\mathfrak{h}_{i+1}(\hat{v}_{i+1}) - \mathbb{E}^{\mathsf{J}}[\mathfrak{h}_{i+1}(\hat{v}_{i+1})])],\nonumber\\
    (\hat{C}_{i+1}^{\mathfrak{h}\mathfrak{h}})^\mathsf{J} &= \mathbb{E}^{\mathsf{J}} [(\mathfrak{h}_{i+1}(\hat{v}_{i+1}) - \mathbb{E}^{\mathsf{J}}[\mathfrak{h}_{i+1}(\hat{v}_{i+1})])\\
    &\quad\quad\quad\otimes (\mathfrak{h}_{i+1}(\hat{v}_{i+1}) - \mathbb{E}^{\mathsf{J}}[\mathfrak{h}_{i+1}(\hat{v}_{i+1})])]\nonumber.
\end{align}
Furthermore, we introduce the following 
rescalings adopted in Law, Stuart, and Zygalakis (2015)\cite{law_data_2015}:
\begin{equation}
\begin{aligned}
    f(\cdot, t) &= (\Psi{_i}(\cdot) - I\cdot)/\tau,\quad z_{i+1}^\dagger - z_i^\dagger = \tau y_{i+1}^\dagger,\\
    \Sigma(i\tau) &= (\Sigma_d)_i/\tau\quad\Gamma(i\tau) = \tau(\Gamma_d)_i,\quad i = t/\tau,
\end{aligned}
\end{equation}
Then, Eqs.~\eqref{eq:discrete} can be seen to be a discretization of Eqs.~\eqref{eq:firstsub} with time step $\tau$. More justification is given for these rescalings in Salgado, Middleton, and Goodwin (1988)\cite{salgado_connection_1988} and Simon (2006)\cite{simon_optimal_2006}. Note that both $K{_{i+1}} = (\hat{C}_{i+1}^{v\mathfrak{h}})^{\mathsf{J}}\,{(\Gamma_d)_{i+1}^{-1}}$ and Eq.~(\ref{eq:discrete}d) are consistent with the continuous-time gain as $\tau\to 0$. We use the latter, similar to the discrete-time Kalman filter.

\subsection{Score Function Term}\label{ssec:score}

We now discuss further computational issues that arise when Eq.~(\ref{eq:mean_field}a) is replaced by \eqref{eq:mean_field_K2}. This term involves the score function, defined as $\nabla\log\rho$, but with an additional preconditioning. If this term is added to the discrete-time particle version of the filter, Eq.~(\ref{eq:discrete}b) becomes
\begin{align}
    v^{(j)}_{i+1} &= \hat{v}^{(j)}_{i+1} + K_{i+1}(y_{i+1}^\dagger - \hat{y}^{(j)}_{i+1})\\
     &\quad\quad\quad\quad+K_{i+1}(\Gamma_d)_{i+1} K_{i+1}^T (\nabla\log\rho_{i+1})^{\mathsf{J}}(\hat{v}^{(j)}_{i+1}),\nonumber
\end{align}
where $(\nabla\log\rho_{i+1})^{\mathsf{J}}$ denotes particle-based approximation of the score function using $\{ \hat{v}^{(j)}_{i+1}\}_{j=1}^{\mathsf{J}}$. If we make the assumption that the density is Gaussian with mean $\mathbb{E}v$ and covariance $C^{vv}$, the score function takes on a simple form,
\begin{equation}
    \nabla \log\rho = -(C^{vv})^{-1}(v - \mathbb{E}v).\label{eq:score_gaussian}
\end{equation}
A natural particle approximation $(\nabla \log\rho)^\mathsf{J}$ follows by replacing the mean and covariance with the corresponding quantities computed under the empirical measure of the set of particles.

More general kernel-based nonparametric estimators for the score function have been developed, such as those defined in Zhou, Shi, and Zhu (2020)\cite{zhou_nonparametric_2020} and implemented in the {\tt kscore} package. In the numerical experiments reported in this paper, we either omit the score term completely, or use it and employ only the Gaussian approximation.

\subsection{Implementation}\label{ssec:implementation}

\subsubsection{Ensemble Square-Root Formulation}

In order to make the method scale well to high dimensions, an ensemble square-root formulation \citep{tippett_ensemble_2003} of Eq.~\eqref{eq:discrete} can be used, although we do not use it in the numerical experiments reported here. The advantage of this formulation is that the most expensive linear algebra operations are rewritten in the ensemble space, resulting in favorable computational complexity when $J$ is much smaller than the state-space dimension $d$ or observation-space dimension $p$.\footnote{Note, however, that in many applications with a high-dimensional state space, the statistics of interest may be relatively low-dimensional, such that the regular version of the algorithm \eqref{eq:discrete} will be feasible.}

To implement this method, we write $(C^{vv})^{\mathsf{J}} = VV^T$, $(C^{v\mathfrak{h}})^{\mathsf{J}} = VY^T$, and $(C^{\mathfrak{h}\mathfrak{h}})^{\mathsf{J}} = YY^T$, where the $j$th column of $V$ and the $j$th column of $Y$ are given by
\begin{equation}
\begin{aligned}
V^{(j)} &= (v^{(j)} - \mathbb{E}^{\mathsf{J}}v)/\sqrt{J - 1},\\
Y^{(j)} &= (\mathfrak{h}(v^{(j)}) - \mathbb{E}^{\mathsf{J}}\mathfrak{h})/\sqrt{J - 1},
\end{aligned}
\end{equation}
respectively. {Then, $K$ can be written as
\begin{equation}
    K = VY^TW,
\end{equation}
where $W = (\Gamma_d^{-1} - \Gamma_d^{-1}Y(I + Y^T\Gamma_d^{-1}Y)^{-1}Y^T\Gamma_d^{-1})$ by the Woodbury identity.}

We assume that $\Gamma_d^{-1}$ is provided and can be applied cheaply, for example, if it is diagonal. This is a standard assumption \citep{tippett_ensemble_2003}. With this expression, $K$ can be computed in $\mathcal{O}(J^3 + J^2p + Jp^2 + dJp)$.

Note that the Gaussian score function approximation Eq.~\eqref{eq:score_gaussian} cannot be applied in cases when $J < d$ since $(C^{vv})^\mathsf{J}$ will be singular. We do not consider the score function term in the complexity analysis.

The complexity is, thus, a quadratic polynomial in $d$ and $p$, whereas various ensemble square-root filters can be implemented to be linear in $p$ and $d$. The latter rely on the fact that the in the standard Kalman filter the updated covariance can be written as $(I - KH)C^{vv}$, where $H$ is the observation operator. The EnFPF cannot be written in this way. Whether the EnFPF can be reformulated to be linear in $p$ and $d$ by another approach is a topic for future research.

\subsubsection{Code}\label{sssec:code}

The open-source Julia code for the EnFPF is available at \url{https://github.com/eviatarbach/EnFPF}. In the numerical experiments that follow, we compute the Wasserstein distance (explained in Sec. \ref{sec:experiments}) using the Python Optimal Transport library \citep{flamary_pot_2021}. We used the parasweep library to facilitate parallel experiments \citep{bach_parasweep_2021}.

\subsubsection{Numerical Methods for the Test Models}\label{ssec:numerical_methods}

In Sec. \ref{sec:experiments}, we will present numerical experiments with the Lorenz63, Lorenz96, and Kuramoto--Sivashinsky models. We integrate the Lorenz63 and Lorenz96 models using the fourth-order Runge--Kutta method, with a time step of 0.05 for both. We integrate the Kuramoto--Sivashinsky equation in Fourier space using the exponential time differencing fourth-order Runge--Kutta method \citep{kassam_fourth-order_2005} with 64 Fourier modes and a time step of 0.25.

\section{Numerical Experiments}\label{sec:experiments}

In this section, we present the results of numerical experiments applying the discrete-time EnFPF of Sec. \ref{ssec:discrete-time} to the Lorenz63, Lorenz96, and Kuramoto--Sivashinsky systems. The first three subsections are devoted, respectively, to these three models; a final fourth subsection returns to the Lorenz63 model, with quasiperiodic forcing.

We found in the experiments that assimilating too often can cause degraded results for some systems, as opposed to the situation in standard filtering, where increased assimilation frequency is typically preferred. {In the standard filtering problem, there is a single true trajectory, and under certain conditions, the filtering distribution will converge to this trajectory in the limit of zero observational noise\cite{dashti_map_2013,law_data_2015}.} {In the non-zero noise case, however, the filtered time-series (e.g., the maximum \textit{a posteriori} estimate) will not even generally be a trajectory of the dynamical system}, except in methods such as strong-constraint 4DVar. Here, we expect that the problem of ensemble members deviating from being trajectories can be amplified, since the method only aims to match statistical features of the entire ensemble. Thus, if assimilation is done too frequently, then ensemble members may be pushed too far from being trajectories into unphysical or unstable parts of the phase space. In fact, we found the assimilation frequency to be a key tuning parameter. We refer to a single forecast--assimilation step (Eqs.~\eqref{eq:discrete}) as a \emph{cycle}, as is common in the DA literature, and each cycle lasts for $\tau$ time units.

We found, furthermore, that the score term did not consistently improve filtering performance. In the experiments that follow, we omit the score term except in the experiments with the Kuramoto--Sivashinsky system in Sec. \ref{ssec:ks_exp}, where it leads to clear improvements when used, together with the Gaussian approximation, in the form of Eq.~\eqref{eq:score_gaussian}. For both the Lorenz models we found that the inclusion of the Gaussian approximation of the score degraded performance and the use of the kernel-based score approximations, based on the paper of Zhou, Shi, and Zhu (2020)\cite{zhou_nonparametric_2020}, was no better than simply omitting the term altogether.

In the experiments below, we use a {Wasserstein metric to quantify the distance between the ensemble distribution and the invariant density. We estimate the invariant density using an ensemble integrated for a sufficiently long time. We employ the $W_1$ Wasserstein metric that allows us to compute distances between empirical distributions.} 
The code for computing this distance is readily available (see Sec. \ref{sssec:code}).

\subsection{Lorenz63 Model}

For the experiments in this subsection, we use the Lorenz (1963)\cite{lorenz_deterministic_1963} model
\begin{equation}
    \begin{aligned}
        \frac{dx}{dt} &= \sigma(y - x),\\
        \frac{dy}{dt} &= x(r - z) - y,\\
        \frac{dz}{dt} &= xy - \beta z,
    \end{aligned}\label{eq:lorenz63}
\end{equation}
with the standard parameter values $\sigma = 10$, $r = 28$, and $\beta = 8/3$.

\subsubsection{Assimilating Time-Varying Means and Second Moments}\label{sssec:lorenz63_means_vars}

We first verify the ability of the EnFPF to force an ensemble to adopt time-varying statistics. We do this by applying the EnFPF to a 10-member ensemble, with noisy statistical observations of the means and uncentered second moments of the three variables coming from a 100-member ensemble being evolved concurrently. The difference between the statistics computed over the 10- and 100-member ensembles arise due to both sampling errors and different initial conditions. The 100-member ensemble (despite having its own sampling error) better approximates the true statistics of the system, and we view these 100-member ensemble statistics as the truth, based on which we may compute errors in the statistics of 10-member ensembles. We assimilate observations every 0.2 time units, with an observation error covariance set to 20\% of the time variability of each statistic computed over the 100-member ensemble.

\begin{figure}
    \centering
    \includegraphics[scale=0.3]{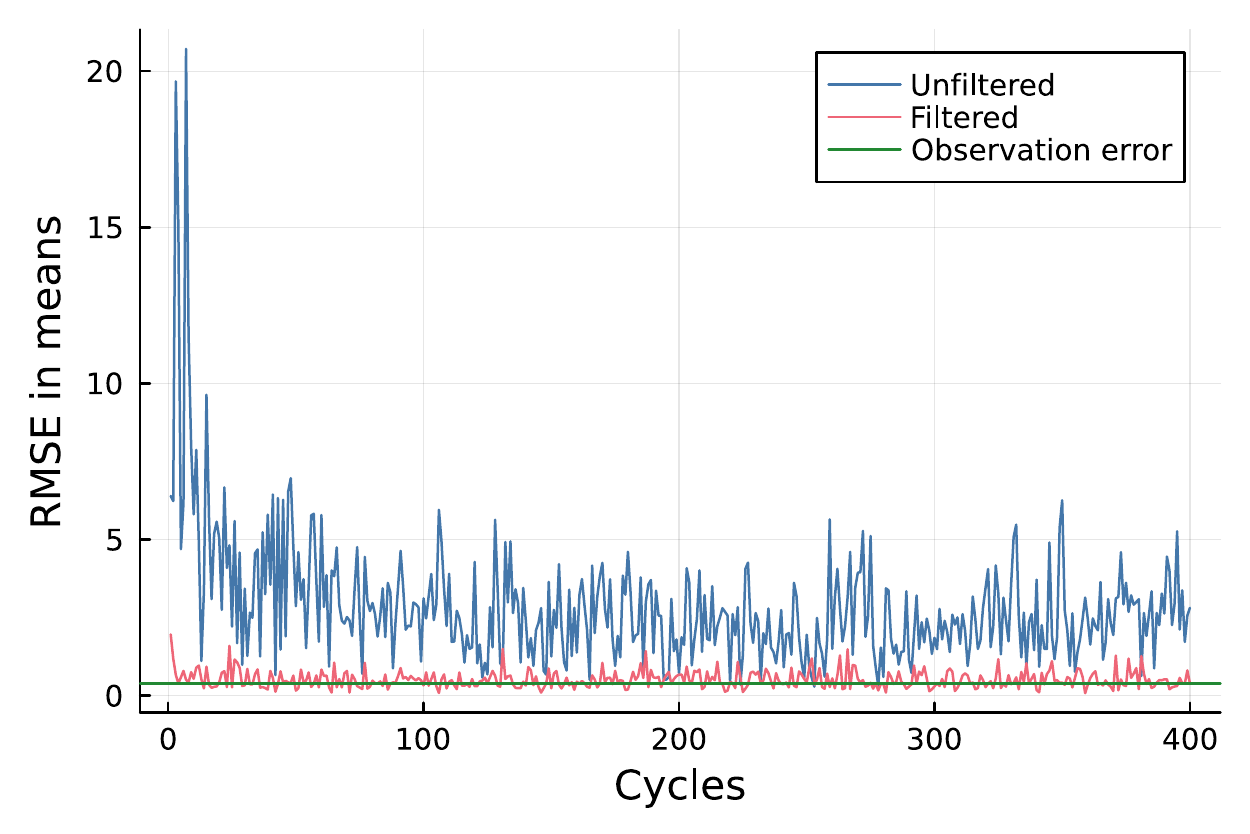}
    \includegraphics[scale=0.3]{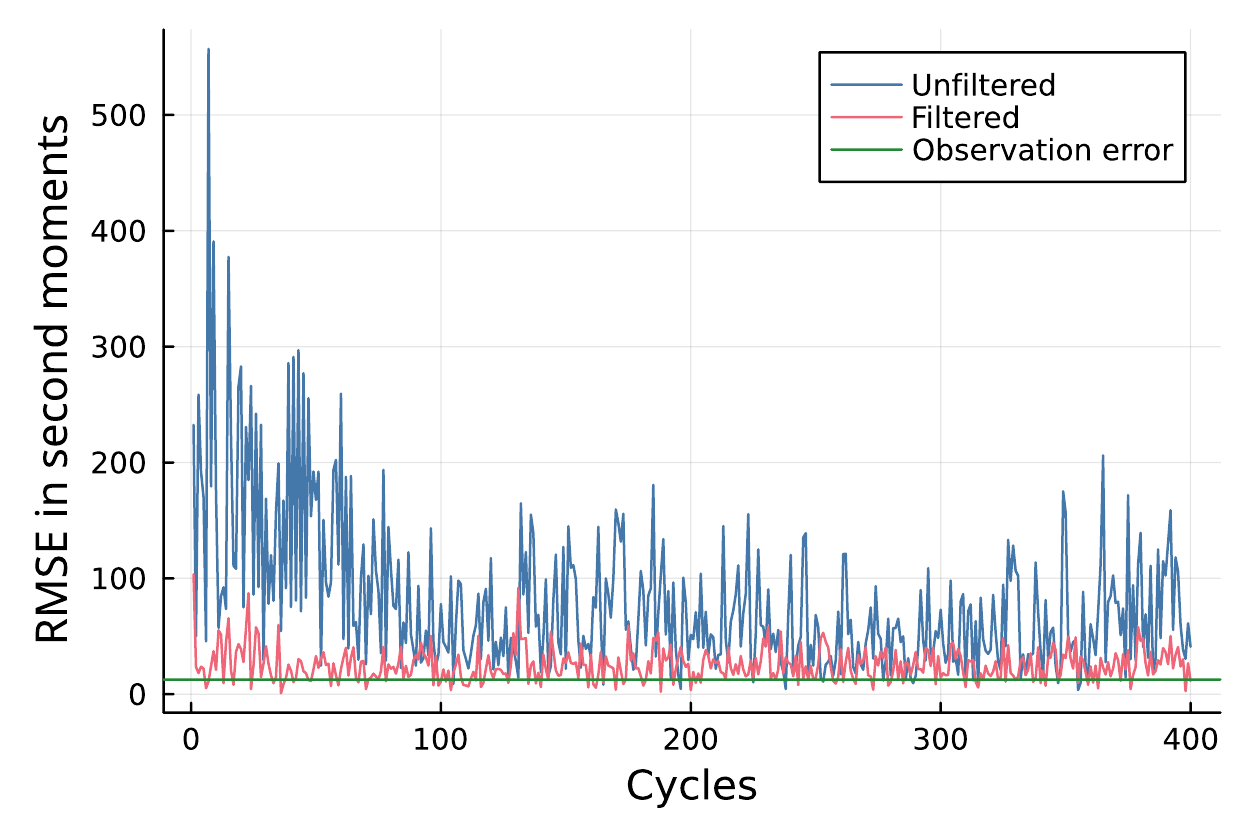}
    \caption{The impact of filtering on the root-mean-square error (RMSE) in the mean and second moment in the Lorenz63 model.}
    \label{fig:lorenz63_means_vars}
\end{figure}

Figure~\ref{fig:lorenz63_means_vars} shows the resulting error in the means and second moments of the 10-member ensemble, compared with the errors arising from an unfiltered run of the 10-member ensemble; in both cases the errors are computed by comparison with the 100-member ensemble. After several cycles, the filter appears to reach an asymptotic error on the order of the observation error, and this error is significantly lower than that arising in the unfiltered case.

\begin{table}
\begin{tabular}{ll}
 \textbf{Means}\\
 \textit{Observation error} & \textit{Filtered RMSE} \\
 \hline
 10\% (0.088) & 0.11\\
 35\% (0.31) & 0.40\\
 60\% (0.53) & 0.69\\
 85\% (0.75) & 0.97\\
\\
 \textbf{Second moments}\\
 \textit{Observation error} & \textit{Filtered RMSE} \\
 \hline
 10\% (2.8) & 20\\
 35\% (9.9) & 23\\
 60\% (17) & 29\\
 85\% (24) & 35\\
\end{tabular}
\caption{The impact of the observation error covariance on filtering performance. In the first column are the percentages of the standard deviation of the time variability of each statistic taken to be the observation error, and in parentheses are the square root of the total variance of the observation error in the statistic. With no filtering, the RMSE is 2.5 in the unfiltered means and 73 in the second moments. The RMSE is averaged over 1400 cycles after 100 transient cycles.}
\label{table:obs_error}
\end{table}

Table~\ref{table:obs_error} shows the impact of the observation error covariance magnitude on the filtering performance. The `setup is otherwise the same as that described above. As expected, the error increases as $\Gamma$ is increased, although still outperforming the unfiltered ensemble.

\subsubsection{Accelerating Convergence to the Invariant Density}

We now test the ability of the EnFPF to accelerate convergence to the invariant density. We assimilate observations of fixed statistics of the invariant density, the means and second moments of the three variables, into a 100-member ensemble. We use the same assimilation frequency and observation error as in Sec. \ref{sssec:lorenz63_means_vars}.

\begin{figure}
    \centering
    \includegraphics[scale=0.3]{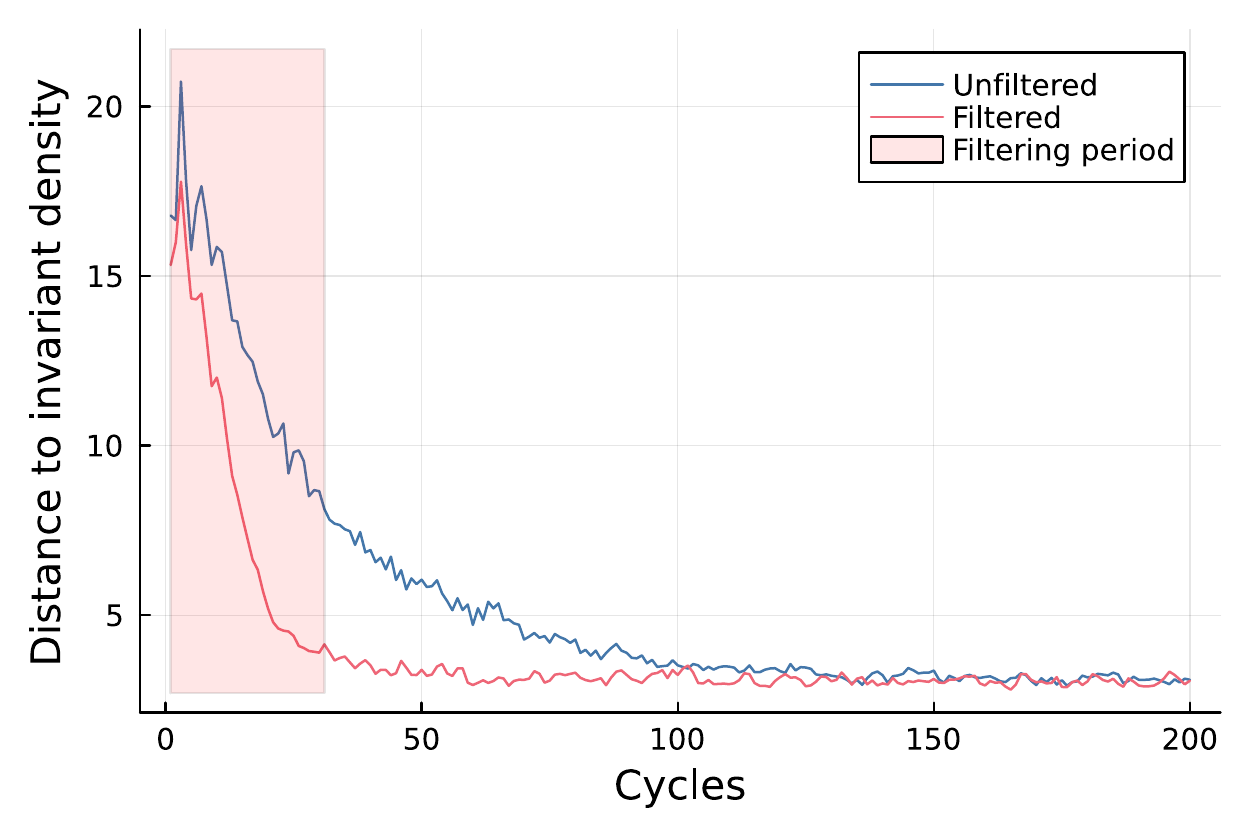}
    \caption{The estimated Wasserstein distance to the invariant density in Lorenz63, in unfiltered and filtered cases. For the filtered case, the first and second moments are assimilated. Each curve is averaged over 10 different initializations.}
    \label{fig:lorenz63_wasserstein}
\end{figure}

Figure~\ref{fig:lorenz63_wasserstein} shows the impact of the EnFPF on the convergence to the invariant density. In this case, we only apply the EnFPF for the first 30 cycles (indicated by the pink rectangle), and then let the ensemble evolve under the regular Lorenz63 dynamics. We see that the EnFPF leads to a more rapid convergence: by the end of the filtering period, the distance is close to the asymptotic one, while it takes at least 100 cycles for the unfiltered case to reach the same.
Figure~\ref{fig:lorenz63_density} visualizes in state space this rapid convergence toward the invariant density via the EnFPF.

\begin{figure*}
    \centering
    \includegraphics[scale=0.75]{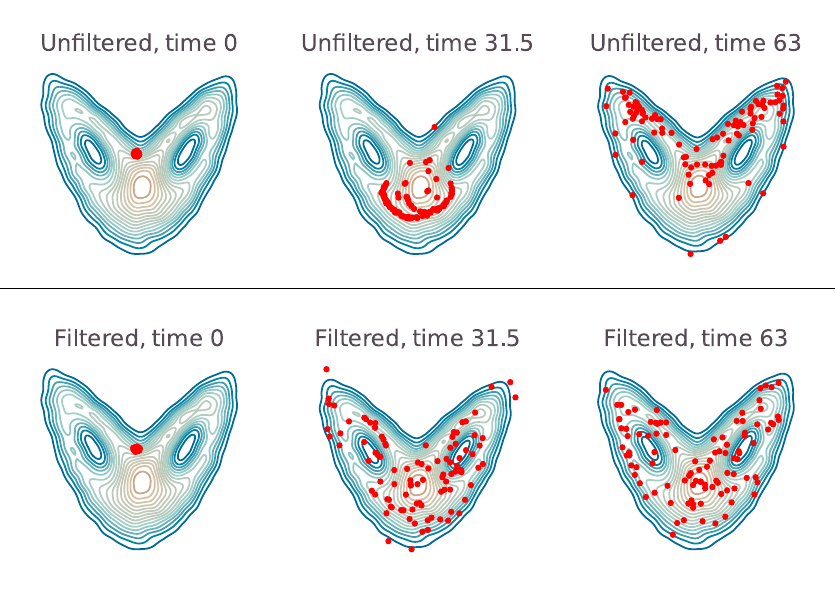}
    \caption{Top panel: an ensemble evolving in time from left to right, superimposed on the invariant density of Lorenz63 in the $x$--$z$ plane. Orange corresponds to higher probability density and blue to lower. Bottom panel: the same but with the EnFPF applied.}
    \label{fig:lorenz63_density}
\end{figure*}

\subsubsection{Impact of Higher-Order Moments}
\begin{figure}
    \centering
    \includegraphics[scale=0.3]{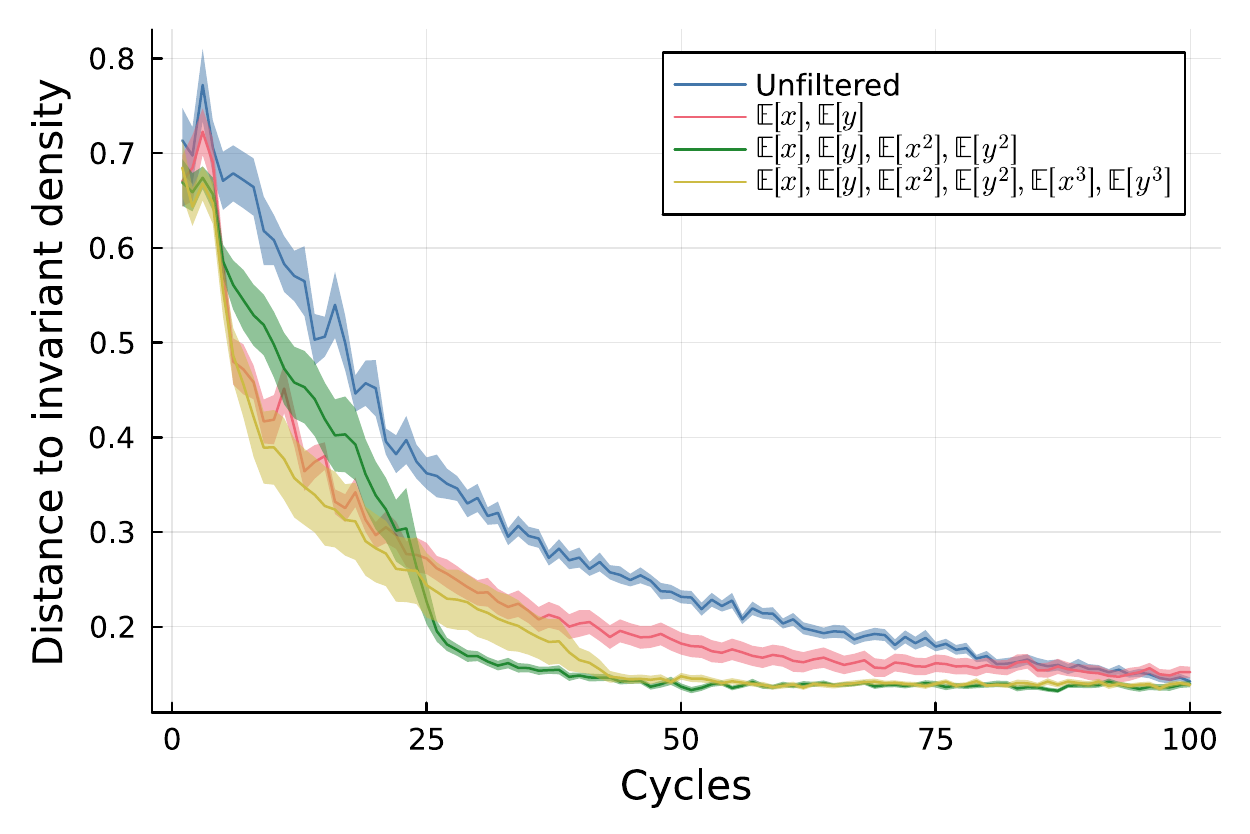}
    \caption{The estimated Wasserstein distance to the invariant density in Lorenz63, in unfiltered and filtered cases when different moments are assimilated. The curves are averaged over 25 initial conditions, and the shaded areas correspond to $\pm$ the standard error over the initializations. Here, for the filtered cases, the EnFPF is applied at every cycle.}
    \label{fig:lorenz63_partial}
\end{figure}

Figure~\ref{fig:lorenz63_partial} shows the convergence to the invariant measure of Lorenz63 with different assimilated moments of $x$ and $y$, namely the first, first and second, and first, second, and third marginal moments. Assimilating the first-order moments accelerates the convergence to the invariant measure compared to the unfiltered case. Adding the second- and third-order moments appears to result in the most rapid initial rate of convergence, and after about 50 cycles assimilating the first two and the first three moments leads to a similar asymptotic distance to the invariant measure. 

\subsection{Lorenz96 Model}

We now test the convergence to the invariant density of the Lorenz (1996)\cite{lorenz_predictability:_1996} model
\begin{equation}
\frac{\text{d}x_i}{\text{d}t} = -x_{i-1}(x_{i-2} + x_{i+1}) - x_i + F,
\end{equation}
where the indices $i$ range from 1 to $D$ and are cyclical. We use $F = 8$ and $D = 40$ variables. This is a model of an atmospheric latitude circle that is {commonly used in data assimilation experiments.}

\begin{figure}
    \centering
    \includegraphics[scale=0.3]{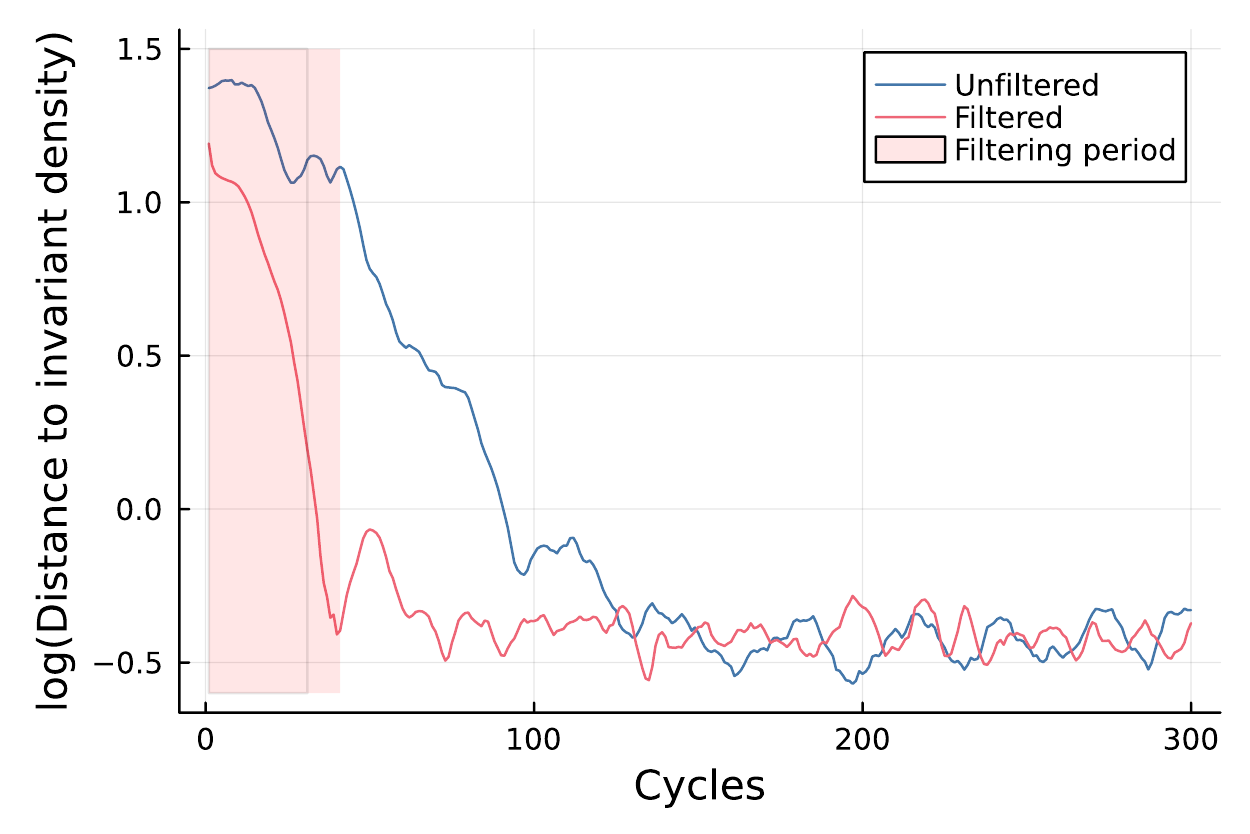}
    \caption{The estimated Wasserstein distance to the invariant density in Lorenz96, in unfiltered and filtered cases. For the filtered case, the first and second moments are assimilated. Here, we show the mean of the Wasserstein distances corresponding to the marginal density for each variable.}
    \label{fig:lorenz96}
\end{figure}

We assimilate the means and second moments of the 40 variables on the invariant density, with an observation error covariance of 20\% of the temporal variability of the statistics computed over a 100-member ensemble. We assimilate every 0.05 time units into a 100-member ensemble for 40 cycles. Figure~\ref{fig:lorenz96} shows that the convergence toward the invariant density is thereby significantly accelerated.

\subsection{Kuramoto--Sivashinsky Model}\label{ssec:ks_exp}

We now carry out experiments with the Kuramoto--Sivashinsky model, a chaotic partial differential equation in one spatial dimension,
\begin{equation}
    u_t + u_{xxxx} + u_{xx} + u u_x = 0 , \quad x \in [0, L].
\end{equation}
We use $L = 22$ and periodic boundary conditions, discretized using 64 Fourier modes (see Sec. \ref{ssec:numerical_methods} for details on the numerical method).

\begin{figure}
    \centering
    \includegraphics[scale=0.3]{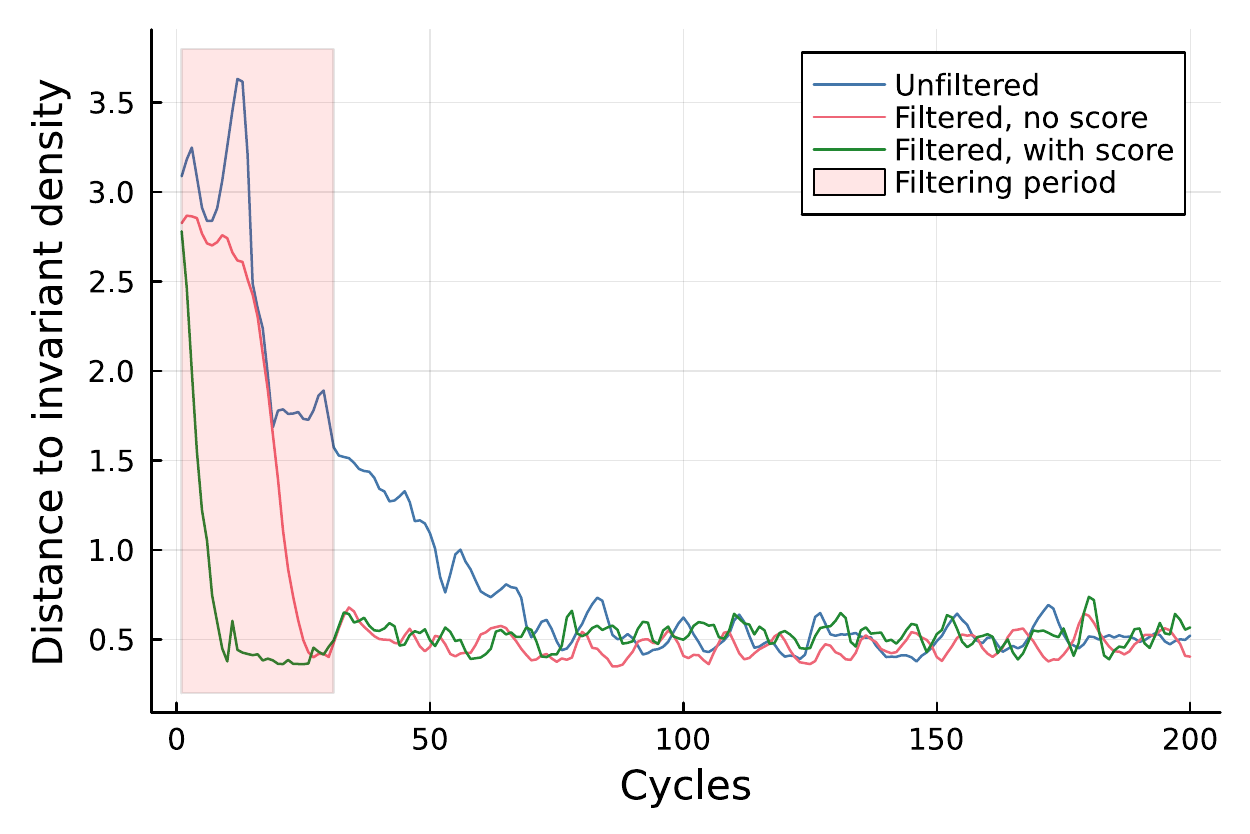}
    \caption{The estimated Wasserstein distance to the invariant density in the Kuramoto--Sivashinsky equation, in unfiltered and filtered cases. For the filtered case, the first and second moments are assimilated. Here, we show the mean of the Wasserstein distances corresponding to the marginal density for each variable.}
    \label{fig:ks}
\end{figure}

We assimilate the means and second moments of the invariant density of the 64 variables in physical space, every 2.0 time units. We assimilate for 30 cycles using a 100-member ensemble, and again use an observational error covariance of 20\% of the temporal variability. Figure~\ref{fig:ks} shows the results with and without the score term included. In both cases, there is an acceleration compared to the unfiltered case; inclusion of the score term considerably accelerates convergence.

\subsection{Time-Dependent Invariant Measures}

We now use the Lorenz63 model (Eq.~\eqref{eq:lorenz63}) but with the $r$ parameter subject to quasiperiodic forcing, as in Daron and Stainforth (2015)\cite{daron_quantifying_2015},
\begin{equation}
r(t) = 28 + \sin \left(2 \pi t\right)+\sin\!\left(\!\sqrt{3} t\right)+\sin\!\left(\!\sqrt{17} t\right).
\end{equation}
Since this system is non-autonomous, it possesses for each time $t$ a pullback attractor with a corresponding time-dependent invariant measure, as discussed in Sec. \ref{ssec:motivation}. The measure at time $t$ can be approximated by the empirical density at time $t$ of an ensemble initialized sufficiently far back in time, at $t - T$ for some large $T$. Here, we evolve a 100-member ensemble using $T=500$ time units to approximate the invariant measures at time $t$. Then, we evolve the ensemble for the additional time period of $t$ to $t + 20$ to obtain approximations to the invariant measures in this period.

We evolve two separate 100-member ensembles for the same time period $t$ to $t+20$, but with $T=0$ (no spin-up). We apply the EnFPF to one of these ensembles and not the other. For the EnFPF, we assimilate every 0.05 time units with an observation error covariance of 20\% of the temporal variability. We then measure the distance between the empirical densities of these two ensembles and the one approximating the invariant measure described in the previous paragraph.

\begin{figure}
    \centering
    \includegraphics[scale=0.3]{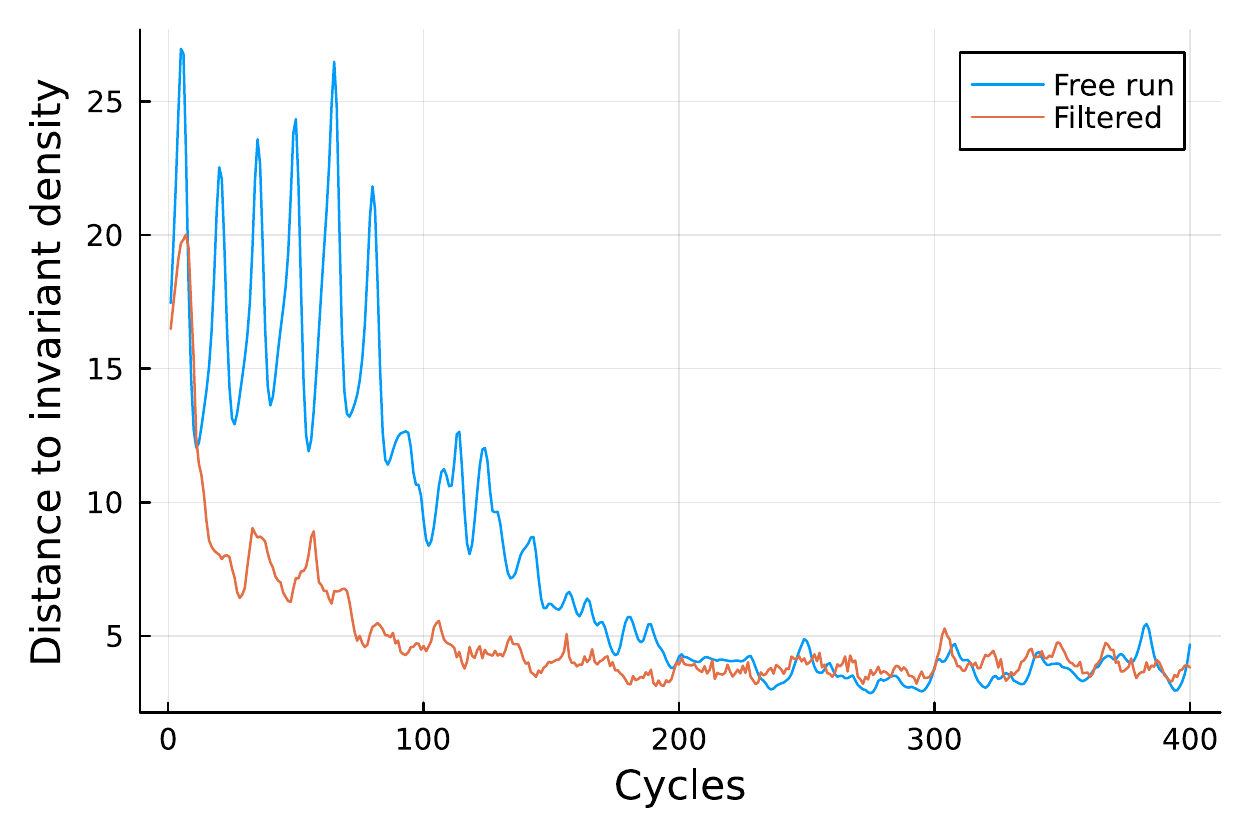}
    \caption{The estimated Wasserstein distance to the invariant density in the non-autonomous Lorenz63 model, in unfiltered and filtered cases. For the filtered case, $\mathbb{E}[x^i]$, $\mathbb{E}[y^i]$, and $\mathbb{E}[z^i]$ for $i=1,2,3$ are assimilated.}
    \label{fig:lorenz63_na}
\end{figure}

Figure~\ref{fig:lorenz63_na} shows that the convergence to the time-dependent invariant measures is indeed accelerated by the EnFPF, reaching a comparable asymptotic distance to the invariant measure in less than half the time.

\section{Justification of Algorithm}\label{sec:justify}

\subsection{Kalman--Bucy (KB) Filter for Densities}\label{ssec:kb}
Since both the Fokker--Planck equation \eqref{eq:fp} and the observation equation \eqref{eq:fp_obs} are linear, and since all noise is additive Gaussian, the conditional probability measure over densities, $\rho | Z^\dagger(t)$, is a Gaussian. This filtering problem can be solved using a Kalman--Bucy filter in Hilbert space, posing significant challenges because it involves finding a sequence of probability measures on an infinite-dimensional space of functions (densities).

We start by defining the Hilbert space $\mathcal{H} = L^2(\mathbb{R}^d, \mathbb{R})$ with inner product
\begin{equation}
    \langle a, b\rangle_{\mathcal{H}}\equiv \int a b\, dv.
\end{equation}
We consider density functions $\rho\in\mathcal{H}$, and we require that $\rho(v, t)\to 0$ as $v\to\infty$. Note that we will sometimes use this inner product in situations where one of the arguments is only locally square integrable; in particular we will need to use the constant function $\mathbbm{1}(v) = 1$. To distinguish them from the Hilbert space inner product, we denote the standard Euclidean inner product in $\mathbb{R}^p$ as $\langle\cdot,\cdot\rangle_{\mathbb{R}^p}$ and the weighted Euclidean inner product, defined for any strictly positive-definite {and symmetric $A \in \mathbb{R}^{p \times p}$,} as $\langle\cdot,\cdot\rangle_A\equiv \langle A^{-1/2}\cdot,A^{-1/2}\cdot\rangle_{\mathbb{R}^p}$.

Recall definition Eq.~(\eqref{eq:fp}b) of the adjoint of the generator $\mathcal{L}$. We are given the dynamics and observation equations \eqref{eq:fp} and \eqref{eq:fp_obs}:
\begin{align}
    d\rho^\dagger(v, t) &= \mathcal{L}^*(t)\rho^\dagger(v, t)\, dt,\\
    dz^\dagger(t) &= H(t)\rho^\dagger(v, t)\, dt + \sqrt{\Gamma(t)}dB.
\end{align}
Then, given all observations up to time $t$, $Z^\dagger(t) = \{z^\dagger(s)\}_{s\in[0,t]}$, the filtering distribution is given by
\begin{equation}
    \rho(\cdot, t) | Z^\dagger(t)\sim \mu(t) \equiv \mathcal{N}\bigl(m(t), C(t)\bigr),
\end{equation}
where $\mathcal{N}$ is a Gaussian measure on $\mathcal{H}$ with mean $m(t)$ and covariance operator $C(t)$. For notational simplicity, we have dropped the explicit dependence of $m(t)$, $C(t)$, and $\rho(t)$ on $v$. Here $C\in L(\mathcal{H}, \mathcal{H})$ is necessarily self-adjoint and trace class \cite{bogachev_gaussian_1998}; that is, $\operatorname{tr}(C) < \infty$. In what follows, the expectation $\mathbb{E}_\mu$ is defined with respect to the measure $\mu$ on the space of $L^2$ densities $\rho$.

Using Theorem 7.10 in Falb (1967)\cite{falb_infinite-dimensional_1967}, the KB filter for this system can be written as
\begin{subequations}
\begin{align}
    dm(t) &= \mathcal{L}^*(t)m(t)\,dt\label{eq:kb_mean}\\
    &\quad\quad\quad+ C(t) H^*(t)\Gamma(t)^{-1}\bigl(dz^\dagger(t) - H(t)m(t)\bigr)\,dt,\nonumber\\
    dC(t) &= \mathcal{L}^*(t) C(t)\, dt + C(t) \mathcal{L}(t)\, dt\label{eq:kb_cov}\\
    &\quad\quad\quad\quad\quad\quad\quad\quad- C(t)H^*(t)\Gamma(t)^{-1}H(t)C(t)\,dt,\nonumber\\
    m(0) &= m_0, C(0) = C_0,\label{eq:kb_init}
\end{align}
\end{subequations}
where
\begin{equation}
    C(t) = \operatorname{cov}({\rho(t)} - m(t), {\rho(t)} - m(t)),
\end{equation}
and
\begin{equation}
    \operatorname{cov}(x, y) \equiv \mathbb{E}_\mu[x\otimes y] - \mathbb{E}_\mu[x]\otimes \mathbb{E}_\mu[y].
\end{equation}
The outer-product $x_1\otimes y_1$ is defined by the identity
\begin{equation}
    (x_1\otimes y_1)x = x_1\langle y_1, x\rangle_{\mathcal{H}}\label{eq:def_circ}
\end{equation}
holding for all $x\in\mathcal{H}$. Note that Falb (1967)\cite{falb_infinite-dimensional_1967} requires boundedness of $\mathcal{L}^*$, but the results have been extended to unbounded operators \citep{curtain_survey_1975}. However, we still require boundedness of $H$. For the rest of the paper, we will assume that $H$ takes the form in Eq.~\eqref{eq:H_def}.

The adjoint operator $H^*$ is then given by
\begin{equation}
    H^*(t) u = \langle\mathfrak{h}(v, t), u\rangle_{\mathbb{R}^p},
\end{equation}
for $u\in\mathbb{R}^p$. Note that, formally, $H^*(t) u$ is to be viewed as a function of $v$, in the space $\mathcal{H}.$

In general, the solution of Eq.~\eqref{eq:kb_mean}, $m(t)$, will not be normalized. However, in Appendix~\ref{sec:kb_properties} we show that normalization is preserved under certain conditions on the initializations $m_0$ and $C_0$ from Eq.~\eqref{eq:kb_init}. However, $m(t)$ is not guaranteed to be non-negative for all $v$ and $t$ and, thus, cannot be a proper probability density. Nonetheless, we can still consider integrals against it.

\subsection{Ansatz and Relation to KB Filter for Densities}\label{ssec:ansatz}

Solving the KB filter equations directly is intractable. We therefore seek an equation that is amenable to a mean-field model, which, in turn, can be approximated by a particle system. We propose the following ansatz for the density of $v|Z^\dagger(t)$:
\begin{equation}
    \frac{\partial \rho}{\partial t} = \mathcal{L}^*{(t)}\rho + \Bigl\langle \mathfrak{h}(v, {t}) - H{(t)}\rho, \frac{dz^\dagger}{dt} - H{(t)}\rho \Bigr\rangle_{\Gamma{(t)}}\rho.\label{eq:ansatz}
\end{equation}

Note the similarity to the Kushner--Stratonovich (KS) equation \eqref{eq:ks}. Although the solutions to this equation do not match the KB filter for densities in general, we show in Theorem~\ref{theorem:equiv} that they coincide in observation space for linear $f$ and $\mathfrak{h}$, under additional assumptions detailed there. The proof sketch is provided in Appendix \ref{appendix:proof}.


\subsection{Mean-Field Approximation}\label{sec:mean_field_intro}

We would now like to find a mean-field model which has, as its FP equation, Eq.~\eqref{eq:ansatz}. We postulate the following form:
\begin{align}
    dv &= f(v, {t}) \, dt + \sqrt{\Sigma{(t)}}\, dW + a(v, \rho, {t})\,dt\label{eq:mean_field_post}\\
    &\qquad + K(v, \rho, {t})\Bigl(dz^\dagger - H{(t)}\rho{(v, t)}\, dt - \sqrt{\Gamma{(t)}} dB\Bigr).\nonumber
\end{align}
Specifically, we aim to choose the pair of functions $(a,K)$ so that the Fokker--Planck equation for $v$ governed by this mean-field model coincides with Eq.~\eqref{eq:ansatz}. In Appendix~\ref{ssec:mean-field}, we detail the choices which achieve this, and after making a further approximation of $K$, we obtain equations \eqref{eq:mean_field} with (\ref{eq:mean_field}a) replaced by \eqref{eq:mean_field_K2}. However, as explained there, 
{in many cases use of Eq.~(\ref{eq:mean_field}a),} which corresponds to setting $a \equiv 0$ and using a simple approximation of $K$, leads to algorithms that empirically perform well.

\section{Conclusions}\label{ssec:conclusions}

In this paper we introduce the Fokker--Planck filtering problem, which consists of estimating the evolving probability density of a (possibly stochastic) dynamical system given noisy observations of expectations evaluated with respect to it. We provide a solution to this problem using the KB filter in Hilbert space, and introduce an ensemble algorithm, the ensemble Fokker--Planck filter (EnFPF), that approximates it under conditions on the dynamics and observables. We also show, through numerical experiments, that this method can be used to accelerate convergence to the invariant measure of dynamical systems, and that this acceleration phenomenon applies beyond the conditions on the dynamics and observables required to provably link the KB filter and the mean-field model underlying our proposed ensemble method.

Future work will test this method on higher-dimensional models, such as turbulent channel flows and ocean models. Other future directions, as described in Sec. \ref{ssec:motivation}, include (i) the testing of this method as an approach to counteract model error, (ii) use in  parameter estimation, and (iii) use in the acceleration of sampling methods such as Langevin dynamics and Markov chain Monte Carlo when some statistics of the target density are known. Furthermore, many of the numerical results require deeper understanding; these include the impact of the assimilation frequency, the score term, and the incorporation of higher-order moments, or other observables, on the filter performance. Finally, on the theoretical side, there is a considerable need for deeper analysis.

\section*{Acknowledgments}
EB is supported by the the Foster and Coco Stanback Postdoctoral Fellowship. AS is supported by the Office of Naval Research (ONR) through grant N00014-17-1-2079. TC and AS acknowledge recent support through ONR grant N00014-23-1-2654. EB and AS are also grateful for support from the Department of Defense Vannevar Bush Faculty Fellowship held by AS. We thank Tapio Schneider, {Dimitris Giannakis}, and two anonymous referees for helpful comments.

\bibliography{references}

\appendix

\section{Properties of the KB Filter for Densities}\label{sec:kb_properties}

Lemma \ref{lemma1} and Remark \ref{remark:normalization} below give the conditions under which $m(t)$ and $\rho(t)\sim\mathcal{N}(m(t), C(t))$ will be normalized. The function $\mathbbm{1}$ is defined as $\mathbbm{1}(v) \equiv 1$ for all $v$.

\begin{lemma}\label{lemma1}
Assume that $\rho(0)\sim \mu(0) = \mathcal{N}(m_0, C_0)$ with
\begin{equation}
    \begin{cases}
        \langle m_0, \mathbbm{1}\rangle_{\mathcal{H}} = 1,\\
        C_0\mathbbm{1} = 0.
    \end{cases}\label{eq:m0C0}
\end{equation}
Then, for $m(t)$ and $C(t)$ satisfying equations \eqref{eq:kb_mean}--\eqref{eq:kb_init},
\begin{enumerate}
    \item[(a)] $C(t)\mathbbm{1} = 0$ for all $t\geq 0$, and
    \item[(b)] $\langle m(t), \mathbbm{1}\rangle_{\mathcal{H}} = 1$ for all $t\geq 0$.
\end{enumerate}
\end{lemma}
\begin{proof}(Sketch)
\begin{itemize}
    \item[(a)] Since $\mathcal{L}\mathbbm{1} = 0$, we have
    \begin{equation}
    \frac{d}{dt}(C\mathbbm{1}) = \mathcal{L}^*C\mathbbm{1} - CH^*\Gamma^{-1}HC\mathbbm{1}.\label{eq:C1}
\end{equation} Assuming uniqueness of the solution to Eq.~\eqref{eq:kb_cov} for the evolution of $C(t)$, we deduce that $C(t)\mathbbm{1} = 0$ solves Eq.~\eqref{eq:C1}.
    \item[(b)] Applying It\^o's lemma to $\langle m, \mathbbm{1}\rangle_{\mathcal{H}}$ (the It\^o correction does not appear due to linearity of the inner product),
    \begin{equation}
    \begin{aligned}
        \frac{d}{dt}\langle m, \mathbbm{1}\rangle_{\mathcal{H}} &= \langle\mathcal{L}^* m, \mathbbm{1}\rangle_{\mathcal{H}} + \left\langle CH^*\Gamma^{-1}(dz^\dagger - Hm), \mathbbm{1}\right\rangle_{\mathcal{H}},\label{eq:rho1}\\
        &= \langle m, \mathcal{L}\mathbbm{1}\rangle_{\mathcal{H}} + \left\langle H^*\Gamma^{-1}(dz^\dagger - Hm), C\mathbbm{1}\right\rangle_{\mathcal{H}},\\
        &= 0,
    \end{aligned}
    \end{equation}
    since $\mathcal{L}\mathbbm{1} = 0$, $C$ is self-adjoint by construction,
    and $C\mathbbm{1} = 0$ by (a). Now assuming uniqueness of the equation~\eqref{eq:kb_mean} for $m(t)$ we find that 
    $\langle m(t), \mathbbm{1}\rangle_{\mathcal{H}} = 1$ solves Eq.~\eqref{eq:rho1}.
\end{itemize}
\end{proof}

\begin{remark}\label{remark}\label{remark:normalization}
    If the conditions in Eq.~\eqref{eq:m0C0} hold then $\langle\rho(t), \mathbbm{1}\rangle = 1$ for $t\geq 0$ almost surely, where $\rho(t)\sim\mu(t)=\mathcal{N}(m(t), C(t))$. This is because
    $\mathbbm{1}$ is in the null-space of both the symmetric operator square-root of $C(t)$, $\sqrt{C(t)}$, and
    \begin{equation}
        \rho(t) = m(t) + \sqrt{C(t)}\xi,
    \end{equation}
    where $\xi\sim\mathcal{N}(0, \mathbb{I})$, with $\mathbb{I}$ being the identity. Thus
    \begin{equation}
    \begin{aligned}
        \langle\rho(t), \mathbbm{1}\rangle_{\mathcal{H}} &= \langle m(t), \mathbbm{1}\rangle_{\mathcal{H}} + \langle\sqrt{C(t)}\xi, \mathbbm{1}\rangle_{\mathcal{H}},\\
        &= 1 + \langle\xi, \sqrt{C(t)}\mathbbm{1}\rangle_{\mathcal{H}},\\
        & = 1.
    \end{aligned}
    \end{equation}
This explains the importance of the conditions in Eq.~\eqref{eq:m0C0}: they ensure that $\rho(t)$ is normalized.
\end{remark}

\section{Theorem~\ref{theorem:equiv}}\label{appendix:proof}

\begin{theorem}
    Assume that:
    \begin{enumerate}
        \item The system dynamics $f$ and $\mathfrak{h}$ are linear in state space: $f(v, {t}) = \mathsf{L}^T v$ and $\mathfrak{h}(v, {t}) = \mathsf{H} v$, with injective $\mathsf{H}$.
        \item $\Sigma = 0$.
        \item $\rho(0)$ is chosen such that its mean $\mathsf{m}(0)$ and covariance $\mathsf{C}(0)$ satisfy 
        \begin{equation}
            \begin{cases}
                \mathsf{H} \mathsf{m}(0) = Hm_0,\\
                \mathsf{H}\mathsf{C}(0)\mathsf{H}^T = HC_0H^*.
            \end{cases}\label{eq:theorem_ic}
        \end{equation}
        \item $m(t)$ stays in the subspace 
        \begin{equation*}
            \mathcal{S} \equiv \left\{u\in \mathcal{H} \,\Bigl| \int |u(v)|v_i v_j dv < \infty\;\forall i, j \in \{1,\ldots, d\}\right\}
        \end{equation*} and $C(t)$ stays in $L(\mathcal{S}, \mathcal{S})$, the space of bounded linear operators from $\mathcal{S}$ into itself.
    \end{enumerate}

    Then, under the same noise realization for $Z^\dagger$, $\mathsf{H} \mathsf{m}(t) = Hm(t)$ and $\mathsf{H}\mathsf{C}(t)\mathsf{H}^T = HC(t)H^*$ will hold for $t\geq 0$, where $\mathsf{m}(t)$ and $\mathsf{C}(t)$ are the mean and covariance of $\rho(t)$ obtained from Eq.~\eqref{eq:ansatz}, and $m(t)$ and $C(t)$ are given by the KB filter for densities \eqref{eq:kb_mean}--\eqref{eq:kb_init}.
\label{theorem:equiv}
\end{theorem}

\begin{proof}(Sketch)

We give here the outlines of a proof, but a rigorous proof, as well as analysis of whether the equivalence holds in any setting more general than the above restrictive conditions, will require considerably more work.

We consider the evolution of the mean and covariance of the KB filter for densities (Eqs.~\eqref{eq:kb_mean} and \eqref{eq:kb_cov}) projected into observation space,
\begin{subequations}
\begin{align}
    d(Hm) &= H\mathcal{L}^*m\,dt + H C H^*\Gamma^{-1}(dz^\dagger - Hm\,dt),\label{eq:dH_rho}\\
    d(HCH^*) &= H\mathcal{L}^* C H^*\,dt + HC \mathcal{L}H^*\, dt\nonumber\\
    &\quad\quad\quad\quad\quad\quad- HCH^*\Gamma^{-1}HCH^*\,dt,\label{eq:dHCH*}
\end{align}
\end{subequations}
where $H(t) = H$ is not time-dependent because $\mathfrak{h}(v, t) = \mathfrak{h}(v) = \mathsf{H}v$. These equations now describe the time evolution of the finite-dimensional quantities $Hm$ and $HCH^*$.

Now, imposing $f(v, {t}) = \mathsf{L}^T v$ and $\mathfrak{h}(v, {t}) = \mathsf{H} v$ on the ansatz (Eq.~\eqref{eq:ansatz}), the time evolution of $\rho$ can be entirely characterized by its mean and covariance, and we obtain the following equations for them:
\begin{subequations}
\begin{align}
    d\mathsf{m} &= \mathsf{L}^T \mathsf{m}\,dt + \mathsf{C} \mathsf{H}^T\Gamma^{-1}(dz^\dagger - \mathsf{H}\mathsf{m}\,dt),\\
    d\mathsf{C} &= \mathsf{L}^T \mathsf{C}\,dt + \mathsf{C} \mathsf{L}\,dt - \mathsf{C} \mathsf{H}^T\Gamma^{-1} \mathsf{H} \mathsf{C}\,dt,\label{eq:dP}
\end{align}
\end{subequations}
where $\mathsf{m} \equiv \mathbb{E}[v]$ and $\mathsf{C}\equiv \mathbb{E}[(v - \mathsf{m})(v - \mathsf{m})^T]$. A similar calculation is made in, e.g., section 7.4 of Jazwinski (1970)\cite{jazwinski_stochastic_1970}. In observation space, we have that
\begin{subequations}
\begin{align}
    d(\mathsf{H} \mathsf{m}) &= \mathsf{H} \mathsf{L}^T \mathsf{m}\,dt + \mathsf{H} \mathsf{C} \mathsf{H}^T\Gamma^{-1}(dz^\dagger - \mathsf{H}\mathsf{m}\,dt),\label{eq:dBv}\\
    d(\mathsf{H} \mathsf{C} \mathsf{H}^T) &= \mathsf{H} \mathsf{L}^T \mathsf{C} \mathsf{H}^T\,dt + \mathsf{H} \mathsf{C} \mathsf{L} \mathsf{H}^T\,dt\nonumber\\
    &\quad\quad\quad\quad\quad\quad- \mathsf{H}\mathsf{C} \mathsf{H}^T\Gamma^{-1} \mathsf{H} \mathsf{C} \mathsf{H}^T\,dt\label{eq:dBPB}.
\end{align}
\end{subequations}

We would now like to show that $\mathsf{H} \mathsf{m}(t) = Hm(t)$ and $\mathsf{H}\mathsf{C}(t)\mathsf{H}^T = HC(t)H^*$ for all $t\geq 0$. We do this by showing that the RHS of Eqs.~\eqref{eq:dH_rho} and \eqref{eq:dHCH*} are equal to the RHS of Eqs.~\eqref{eq:dBv} and \eqref{eq:dBPB} at time $t^*$ if $\mathsf{H} \mathsf{m}(t^*) = Hm(t^*)$ and $\mathsf{H}\mathsf{C}(t^*)\mathsf{H}^T = HC(t^*)H^*$. Together with the initial conditions \eqref{eq:theorem_ic} and uniqueness, this proves the theorem.

It follows immediately that
\begin{equation}
\begin{aligned}
    &H C(t^*) H^*\Gamma^{-1}\!\left[\frac{dz^\dagger}{dt} - Hm(t^*)\right]\\
    &\qquad\qquad= \mathsf{H} \mathsf{C}(t^*) \mathsf{H}^T\Gamma^{-1}\!\left[\frac{dz^\dagger}{dt} - \mathsf{H}\mathsf{m}(t^*)\right],
\end{aligned}
\end{equation}
and that
\begin{equation}
    HC(t^*)H^*\Gamma^{-1}HC(t^*)H^* = \mathsf{H}\mathsf{C}(t^*) \mathsf{H}^T\Gamma^{-1} \mathsf{H} \mathsf{C}(t^*) \mathsf{H}^T.
\end{equation}

Note that
\begin{equation}
\begin{aligned}
    \mathsf{H}\mathsf{m}(t^*) &= Hm(t^*),\\
    &= \mathsf{H}\int vm(t^*)\,dv,
\end{aligned}
\end{equation}
which implies that
\begin{equation}
    \mathsf{m}(t^*) = \int vm(t^*)\,dv,\label{eq:vm}
\end{equation}
because $\mathsf{H}$ was assumed to be injective.

We proceed with the rest of the terms. For the first term of the RHS of Eq.~\eqref{eq:dH_rho},
\begin{equation}
\begin{aligned}
    H\mathcal{L}^*m &= \mathsf{H} \int v \mathcal{L}^*m\,dv,\\
    &= -\mathsf{H} \int v \nabla\cdot(m f)\,dv,\\\label{eq:HL*_rho}
    &= -\mathsf{H}\mathsf{L}^T\int v \nabla\cdot(m v)\,dv,\\
    &= \mathsf{H}\mathsf{L}^T\int v m\,dv,\\
    &= \mathsf{H}\mathsf{L}^T \mathsf{m},
\end{aligned}
\end{equation}
where the fourth line follows from integration by parts and the last from Eq.~\eqref{eq:vm}. Note that the boundary term in the integration by parts vanishes from assumption 4. Thus,
\begin{equation}
    H\mathcal{L}^*m = \mathsf{H}\mathsf{L}^T \mathsf{m}.
\end{equation}

It remains to show that $H\mathcal{L}^* C(t^*) H^* = \mathsf{H} \mathsf{L}^T \mathsf{C}(t^*) \mathsf{H}^T$. We have that for any $u$,
\begin{equation}
    H C(t^*) H^* u = \mathsf{H}\int v C(t^*) H^*u\,dv = \mathsf{H} \mathsf{C}(t^*) \mathsf{H}^T u.
\end{equation}

Since $\mathsf{H}$ was assumed to be injective,
\begin{equation}
    \int v C(t^*) H^*u \,dv = \mathsf{C}(t^*) \mathsf{H}^T u.\label{eq:int_identity}
\end{equation}

Then, for any $w$,
\begin{equation}
\begin{aligned}
    H\mathcal{L}^* C(t^*) H^*w &= \mathsf{H}\int v \mathcal{L}^*C(t^*)H^* w \,dv\\
    &= -\mathsf{H}\int v\nabla\cdot(C(t^*)H^*w \mathsf{L}^T v)\,dv\\
    &= \mathsf{H} \mathsf{L}^T\int v C(t^*)H^*w\,dv\\
    &= \mathsf{H}\mathsf{L}^T \mathsf{C}(t^*) \mathsf{H}^T w
\end{aligned}
\end{equation}
where the third line follows from integration by parts (with the boundary term vanishing by the same argument as above) and the last line from Eq.~\eqref{eq:int_identity}. Taking the adjoint demonstrates that $HC(v, t^*) \mathcal{L}H^* = \mathsf{H} \mathsf{C}(t^*) \mathsf{L} \mathsf{H}^T$, completing the proof.
\end{proof}

\section{Mean-Field Approximation}\label{ssec:mean-field}

We omit the function arguments until the end of the subsection, for brevity.
Using Eq.~3.30 from Calvello, Reich, and Stuart (2022)\cite{calvello_ensemble_2022}, we know that the FP equation of Eq.~\eqref{eq:mean_field_post} when $f(v, {t}) = 0$ and $\Sigma = 0$ is
\begin{align}
    \frac{\partial \rho}{\partial t} &= -\nabla\cdot(\rho (a-KH\rho)) - \left\langle \nabla\cdot(\rho K^T), \frac{dz^\dagger}{dt}\right\rangle\nonumber\\
    &\quad\quad\quad\quad\quad\quad\quad\quad\quad+ \nabla\cdot\bigl(\nabla\cdot(\rho K \Gamma K^T)\bigr).\label{eq:mean_field_fp}
\end{align}

We now match the terms of Eqs.~\eqref{eq:mean_field_fp} and \eqref{eq:ansatz} to make them equal. By matching the terms involving $dz^\dagger/dt$, we obtain that
\begin{equation}
    \Gamma^{-1}({\mathfrak{h}} - H\rho) \rho = -\nabla\cdot(\rho K^T),\label{eq:K_eq}
\end{equation}
and matching the rest of the terms,
\begin{equation}
    -\rho \langle {\mathfrak{h}} - H\rho, H\rho\rangle_{\Gamma}  = -\nabla\cdot(\rho(a - KH\rho)) + \nabla\cdot(\nabla\cdot(\rho K \Gamma K^T)).\label{eq:matched}
\end{equation}
Substituting Eq.~\eqref{eq:K_eq} into Eq.~\eqref{eq:matched}, we obtain
\begin{equation}
\begin{aligned}
    \langle \nabla\cdot(\rho K^T), H\rho\rangle &= \nabla\cdot(\rho K H\rho)\\
    &= -\nabla\cdot(\rho(a - KH\rho))\\\
    &\quad\quad\quad+ \nabla\cdot(\nabla\cdot(\rho K \Gamma K^T)).
\end{aligned}
\end{equation}

Setting the term in the divergence to 0, we obtain
\begin{equation}
    a = K\Gamma K^T \nabla \log \rho.
\end{equation}
This is the origin of the score function term discussed in subsection~\ref{ssec:score}.

We propose a test function $\psi(v) = v - \mathbb{E}v$, take the outer product of it with both sides of Eq.~\eqref{eq:K_eq}, and integrate by parts, obtaining the identity
\begin{equation}
    \mathbb{E}K = \mathbb{E}[\psi ({\mathfrak{h}} - H\rho)^T]\Gamma^{-1} = C^{v\mathfrak{h}}\Gamma^{-1},
\end{equation}
where $C^{v\mathfrak{h}}{(t)}\equiv \mathbb{E}[(\mathfrak{h}(v, {t}) - H\rho)(\mathfrak{h}(v, {t}) - H\rho)^T]$.

Fixing the value of the gain $K$ to its expectation (the constant gain approximation discussed in Calvello, Reich, and Stuart (2022)\cite{calvello_ensemble_2022}), we then obtain
\begin{equation}
    K{(t)} = C^{v\mathfrak{h}}{(t)}\Gamma{(t)}^{-1}.
\end{equation}
Thus, the mean-field model is
\begin{align*}
    dv &= f(v, {t}) \, dt + \sqrt{\Sigma{(t)}}\, dW + K{(t)}\Bigl(dz^\dagger - d\hat{z}\Bigr)\\
    &\quad\quad\quad\quad\quad\quad\quad\quad+ K{(t)}\Gamma{(t)} K{(t)}^T \nabla \log \rho{(v, t)}\,dt,\\
    d\hat{z} &= {(\mathbb{E}\mathfrak{h})(t)}\, dt + \sqrt{\Gamma{(t)}} dB,
\end{align*}
which gives Eqs.~\eqref{eq:mean_field}, with (\ref{eq:mean_field}a) replaced by \eqref{eq:mean_field_K2}.
\end{document}